\RequirePackage{fix-cm}
\documentclass[smallextended]{svjour3}       
\smartqed  
\usepackage{graphicx}
\usepackage[T1]{fontenc}
\usepackage{amsmath}
\usepackage{algorithm}
\usepackage{dsfont}
\usepackage[utf8]{inputenc}
\usepackage{amssymb,amsfonts,amsmath}
\usepackage{thm-restate}

\usepackage{todonotes}

\usepackage{qcircuit}
\usepackage{circuitikz}
\usepackage{tikz}
\usetikzlibrary{backgrounds,fit,decorations.pathreplacing,calc}

\newcommand{\ket}[1]{\ensuremath{\left| #1 \right \rangle}}
\newcommand{\bra}[1]{\ensuremath{\left \langle #1 \right |}}

\def\ketbra#1#2{{\vert#1\rangle\!\langle#2\vert}}

\newcommand{\tder}[2][]{\frac{d #1}{d #2}}

\newcommand{\eye}{\mathds{1}}

\usepackage{hyperref}
\definecolor{burgundy}{rgb}{0.5, 0.0, 0.13}
\hypersetup{colorlinks,linkcolor={blue},citecolor={blue},urlcolor={burgundy}}  
\begin{document}

\title{On the Universality of the Quantum Approximate Optimization Algorithm}

\subtitle{ \\ M.~E.~S.~Morales, J.~D.~Biamonte, and Z.~Zimborás}

\author{}


\institute{M.~E.~S.~Morales \at
              Deep Quantum Laboratory, Skolkovo Institute of Science and Technology, 3 Nobel Street, Moscow, Russia 121205 \\
         \email{mauro990@gmail.com}
         \url{http://quantum.skoltech.ru}
            \and  
            J.~D.~Biamonte \at
                Deep Quantum Laboratory, Skolkovo Institute of Science and Technology, 3 Nobel Street, Moscow, Russia 121205 \\
              \email{j.biamonte@skoltech.ru} 
              \url{http://quantum.skoltech.ru}
           \and
           Z.~Zimborás \at
            Wigner Research Centre for Physics of the Hungarian Academy of Sciences;\\ MTA-BME Lendület Quantum Information Theory Research Group;\\ Mathematical Institute, Budapest University of Technology and Economics, Hungary.  \\
         \email{zimboras.zoltan@wigner.mta.hu}
}

\date{}


\maketitle

\begin{abstract}
The quantum approximate optimization algorithm (QAOA) is considered to be one of the most promising approaches towards using near-term quantum computers for practical application. In its original form, the algorithm applies two different Hamiltonians, called the mixer and the cost Hamiltonian, in alternation with the goal being to approach the ground state of the cost Hamiltonian. Recently, it has been suggested that one might use such a set-up as a parametric quantum circuit with possibly some other goal than reaching ground states. From this perspective, a recent work [S.~Lloyd, arXiv:1812.11075] argued that for one-dimensional local cost Hamiltonians, composed of nearest neighbor ZZ terms, this set-up is quantum computationally universal, i.e., all unitaries can be reached up to arbitrary precision. In the present paper, we give the complete proof of this statement and the precise conditions under which such a one-dimensional QAOA might be considered universal.  We further generalize this type of universality for certain cost Hamiltonians with ZZ and ZZZ terms arranged according to the adjacency structure of certain graphs and hypergraphs.

\end{abstract}

\section{Introduction}
\label{sec:intro}

A question in the field of quantum information processing is whether contemporary quantum processors will in the near future be able to solve problems more efficiently than classical computers. 
Combinatorial optimization problems are of special interest, for which a class of algorithms under the name of Quantum Approximate Optimization Algorithm (QAOA) have been proposed \cite{original-qaoa}. QAOA consists of a bang-bang protocol \cite{bang-bang} that is expected to solve hard problems approximately. This procedure involves the unitary evolution under a Hamiltonian encoding the objective function of the combinatorial optimization problem and a second non-commuting mixer Hamiltonian. Since its proposal, QAOA has been extensively studied to understand its performance \cite{Hadfield,chuang2019,hastings2019}, for establishing quantum supremacy results \cite{supremacy} and for solving several optimization problems \cite{constraint,crooks,Marsh2019}.
This algorithm together with others such as the variational quantum eigensolver (VQE) \cite{original-vqe,variational,2019arXiv190304500B} are part of the so called variational hybrid quantum/classical algorithms, combining the computational power of a quantum computer to prepare quantum states with a classical optimizer. These variational algorithms (including QAOA) have shown several advantages such as robustness to noise, yet more study is required to know the limitations in algorithms such as QAOA. Recent work has found limitations in parametrized quantum circuits trained with classical optimizers wherein for large enough problem sizes the algorithms suffer from so called barren plateaus from which exponentially low probability to escape don't allow the algorithms to achieve an optimal result \cite{mcclean2018barren}. The expressive power of parametrized quantum circuits, namely, the set of probability distributions from which a parametrized circuit is capable to sample from, has also been studied \cite{expresiveness2018}. In this paper, we study the capacity of QAOA to perform universal quantum computation in the sense that sequences of QAOA unitaries can approximate arbitrary unitaries (as we will detail below).

A proof-sketch of the computational universality of a class of QAOA quantum circuits has been given in Ref.~\cite{Lloyd-proof}. In our work, we make the statement concerning the universality of QAOA circuits more precise and give the complete proof of universality inspired by this previous work. We also give the conditions under which the proof in Ref.~\cite{Lloyd-proof} applies. In addition to this, we expand and generalize the proof to include QAOA circuits defined by other classes of cost Hamiltonians. Morover, we also discuss cases when universality is not reached, which helps to further advance the understanding of limitations of QAOA. For our proofs, we employ techniques from Lie group theory utilized previously in the context of quantum control  \cite{JURDJEVIC1972313,Zoltan,Zeier,Dirr,Zoltan2} and also in proving universality of different families of gate sets \cite{universal-gates,bremner2004,Zoltan3,Sawicki}. 
In particular, we will make connections with a graph process named zero forcing that was already connected to Lie algebraic controllability questions \cite{Burgarth-control,zeroforce}. Previous works \cite{bang-bang,wang2018,yuehzen2019} have related controllability to QAOA, our work continuous in this direction and reveals that there are more fruitful connections to be made between these topics.  A recent work by one of the present authors \cite{2019arXiv190304500B} proved that an objective function, expressible in terms of local measurements, can be minimized to prepare arbitrary quantum states as output by quantum circuits. The work, however, assumed the existence of universal variational sequences, such as those needed to realize a universal gate set, but did not prove this reachability. Hence the sequences developed here would find further applications therein, as well.

The paper is organized as follows. We provide some background to our work in Sect.~\ref{sec:background}; the QAOA algorithm is introduced together with the notion of universality used in Sect.~\ref{sec:intro-qaoa} and a brief introduction on quantum control and its relation to QAOA in Sect.~\ref{sec:intro-control}. We then proceed to complete the proof of Ref.~\cite{Lloyd-proof} concerning the universality of a 1d QAOA system in Sect.~\ref{sec:universality-1d}. The generalization of the universality proof to other settings is presented in Sect.~\ref{sec:universality-general} and Sect~\ref{sec:hypergraphs}. Finally, we close with the conclusion and outlook in Sect.~\ref{sec:conclusion}.

\section{Background and setting}\label{sec:background}

Here we summarize the background of our work. We briefly introduce the concept of QAOA and give the precise definition of universality which is used in this article. Then we introduce some notation from quantum control and explain how it relates to our proof of the universality of QAOA under certain conditions. 

\subsection{Quantum Approximate Optimization Algorithm }
\label{sec:intro-qaoa}

The quantum approximate optimization algorithm is used to find solutions to combinatorial optimization problems. To introduce the algorithm, we follow the presentation given in \cite{original-qaoa}. A more complete analysis of the algorithm can be found therein.

The algorithm is defined by a Hamiltonian $H_Z$ encoding the objective function $f:\{0,1\}^n \to \mathbb{R}$ of a combinatorial optimization problem which we wish to minimize (or alternatively, maximize). This Hamiltonian is assumed to be diagonal in the computational basis and is denoted as the cost Hamiltonian. There is also a second Hamiltonian $H_X$ denoted as mixer Hamiltonian which does not commute with $H_Z$.

 First, fix an integer $p$ and 2$p$ random angles $\boldsymbol{\gamma}=(\gamma_1, \gamma_2 ... \gamma_p)$, $\boldsymbol{\beta}=(\beta_1,..\beta_p)$. Then, as a subroutine, prepare using a quantum computer an ansatz state 
 \begin{equation}\label{eq:qaoa-ansatz}
     \ket{\boldsymbol{\gamma},\boldsymbol{\beta}}=U(H_X,\beta_p)U(H_Z,\gamma_p)...U(H_X,\beta_1)U(H_Z,\gamma_1)\ket{+}^{\otimes n} .
 \end{equation}
 
 Where $U(H,\alpha) = e^{-i\alpha H}$ and $\ket{+} = \frac{1}{\sqrt{2}}(\ket{0} + \ket{1})$. This ansatz state is then measured in the computational basis which results in a bitstring $z\in \{0,1 \}^n$. We can then evaluate $f(z)$ by sampling enough times from the ansatz state. Then the following expected value can be approximated 
 
 \begin{equation}
F_p(\boldsymbol{\gamma},\boldsymbol{\beta}) = \bra{\boldsymbol{\gamma},\boldsymbol{\beta}}H_Z\ket{\boldsymbol{\gamma},\boldsymbol{\beta}}.
\end{equation}

With a classical optimization algorithm we seek to minimize this expectation value, and thus we update the angles $\boldsymbol{\gamma}=(\gamma_1, \gamma_2 ... \gamma_p)$, $\boldsymbol{\beta}=(\beta_1,..,\beta_p)$ for the next round. We repeat this procedure for several rounds.

The operator $H_X$ is usually defined as
\begin{equation}
    H_X = \sum_{i=1}^{n} X_i,
\end{equation}
where $X_i$ is the usual Pauli matrix acting on the $i$th qubit.

\subsection{Universality of QAOA as a parametrized quantum circuit}\label{subsec:universality-def}

To study universality we need to define what do we mean by it in the context of QAOA. As explained before, QAOA involves a subroutine where a quantum circuit outputs a quantum state. The family of quantum circuits defined by QAOA from a set of angles and a sequence length is given by the product of unitaries in Eq.~\eqref{eq:qaoa-ansatz}. As discussed in \cite{Nest_2007}, universality in the quantum circuit model is related to the possibility of generating arbitrary unitary operations by composition of elementary gates in a gate set. In this sense we can consider for a choice of $H_Z$ and $H_X$ the unitaries $U(H_Z,\alpha)$ and $U(H_X,\beta)$ for any angles $\alpha$, $\beta$ as an elementary gate set. Thus for fixed Hamiltonians $H_Z$, $H_X$ acting on $n$ qubits and $p\in\mathbb{N}_{>0}$ the family of circuits defined by QAOA corresponds to the set of unitaries
\begin{equation}
    \mathcal{C}_{H_Z,H_X}^{p}{=} \big\{U(H_X,\beta_p)U(H_Z,\gamma_p)...U(H_X,\beta_1)U(H_Z,\gamma_1) |  \gamma_j,\beta_j \in [0,2\pi] \big\}, 
\end{equation}
where $U(H,\alpha) = e^{-i\alpha H}$. Thus, we can define
\begin{equation} \label{eq:inf_sequences}
\mathcal{C}_{H_Z,H_X} = \bigcup_{p=1}^{\infty} \mathcal{C}_{H_Z,H_X}^{p}. 
\end{equation}

For a problem size  $n$ and a choice of $H_Z$ and $H_X$ acting on $n$ qubits we say QAOA is universal if any element in the full unitary group $\mathcal{U}(2^n)$ is approximated to arbitrary precision (up to a phase) by an element of $\mathcal{C}_{H_Z,H_X}$.

Note that our definition of universality does not make reference to the sequence length $p$ of Eq.~\eqref{eq:qaoa-ansatz}. Studying the sequence length at which any unitary in $\mathcal{U}(2^n)$ can be approximated for certain choices of Hamiltonians or even for unitaries in a subspace $\mathcal{A}\subseteq \mathcal{U}(2^n)$ may prove useful in tasks such as state preparation \cite{Hsieh2019,Hsieh2019-2}, modifications of QAOA where constrains are included \cite{QAOA-hard-soft} or for understanding the limitations of this algorithm \cite{akshay2019}. It would also be interesting to investigate universality in other variational quantum algorithms, see Ref.~\cite{herasymenko2019} for a recent study in this direction concerning variational quantum eigensolvers.

Finally, let us stress here again that the notion of universality here does not provide an algorithm that finds the solution of the objective function. It just quantifies the reachability properties of QAOA unitary sequences. An analogous notion of universality in classical variational neural networks was given by the universal approximation theorem \cite{Cybenko1989,HORNIK1991251,universal-approx} which states that under some weak assumptions feed-forward neural networks can approximate any continuous function defined on a compact subset of $\mathbb{R}^k$ without giving an algorithm for the approximation. 

\subsection{Quantum control}\label{sec:intro-control}

The Quantum Approximate Optimization Algorithm can be understood as a particular quantum control problem. Hence it will be useful to briefly introduce the concept of reachability within quantum control theory. 

Let us consider a quantum system with a drift Hamiltonian $H_0$, and assume further that one can turn on or off  the Hamiltonians $H_j$ ($j=1, \ldots, n$) with time-dependent coupling-strengths (control functions) $u_j$, and in this way obtain the following time-dependent control Hamiltonian \begin{equation}\label{eq:total-ham}
  H(t) = H_0 + \sum_{j=1}^q u_j(t) H_j  \, .
\end{equation}
The evolution of the (pure) state of a quantum system is then described by the controlled Schr\"odinger equation 
\begin{equation}\label{eq:schrodinger}
i\hbar \frac{d}{dt}\ket{\psi} = H(t) \ket{\psi} \; ,   {\text{ with initial condition}} \; \; \ket{\psi (t=0)} = \ket{\psi_0}.
\end{equation}

 The solution to equation \eqref{eq:schrodinger} can be written using a unitary propagator $\ket{\psi(t)}=U(t)\ket{\psi_0}$, which can be obtained as the solution to the following differential equation

\begin{equation}\label{eq:schr-unitary}
    \tder{t} U(t) = \left( -iH_0 + \sum_{j=1}^{q} -iu_j(t) H_j \right) U(t) \; \; {\text{with}} \; \; U(0) = \eye.
\end{equation}

We want to answer the following question: given a set of control Hamiltonians $\mathcal{P} = \{iH_1, iH_2, ..., iH_q\} $, which unitary propagators can we generate?

We assume that the control functions $u_j$ all belong to a set $\mathcal{F}$ of allowed control functions which correspond to piecewise constant functions, this choice will be relevant for QAOA.  Before delving more into the problem let us make some definitions.

\begin{definition}[Set of reachable unitaries]
Given a quantum system (described by a $d$-dimensional Hilbert space) with drift Hamiltonian $H_0$ and control Hamiltonians $\{H_j \}_{j=1}^q$, define the set of reachable unitaries at time $T>0$ as the set
\begin{equation}
    \mathcal{R}(T)=\{ W \in \mathcal{U}(d) : \exists u \in \mathcal{F}, \exists U(t) \; \text{solution of Eq.~\eqref{eq:schr-unitary}}, U(T,u) = W\},
\end{equation}
\end{definition}
and the set of reachable unitaries are 
\begin{align}
    \mathcal{R}&= \overline{\cup_{T>0} \mathcal{R}(T)} \nonumber \\
    &=\{ W \in \mathcal{U}(d) :   \forall \epsilon > 0 \; \exists T_\epsilon  \;, \exists U_\epsilon \in \mathcal{R}(T_\epsilon) \; \text{such that} \; \|W- U_\epsilon\| \le \epsilon \},
\end{align}
where $\| \cdot \|$ denotes the operator norm.

\begin{definition}[Generated Lie Algebra]
Given a set of Hamiltonians $\mathcal{P} = \{iH_1, iH_2, ..., iH_q\}$, we call the  smallest real Lie algebra $\mathcal{L}$ containing the elements of $\mathcal{P}$  the {\it generated Lie algebra} of $\mathcal{P}$. We will denote the generated Lie algebra as

\begin{equation}
    \mathcal{L} = \langle \mathcal{P} \rangle_{Lie} =  \langle \{iH_1, iH_2, ..., iH_q\} \rangle_{Lie}.
\end{equation}

\end{definition}

\begin{proposition}\label{thm:reachability}
Given a set of Hamiltonian generators $\mathcal{P}$ defining a set of unitary operators according to  Eq.~\eqref{eq:schr-unitary} (without a drift Hamiltonian $H_0$), then the reachable set of unitaries is the following \cite{alessandro2008introduction}
\begin{equation}
    \mathcal{R} = e^{\mathcal{L}} = \{ e^{A_1} e^{A_2} ... e^{A_m} : m \in \mathbb{N}, A_j \in \mathcal{L} \},
\end{equation}
    where $\mathcal{L}$ is the Lie algebra generated by $\mathcal{P}$. Moreover, if the quantum system is finite dimensional, we have that $e^{\mathcal{L}}=\{e^{A} : A \in \mathcal{L} \}$.
\end{proposition}

Proposition \ref{thm:reachability} motivates us to study the Lie algebra generated by a set of Hamiltonians. To understand whether a set of Hamiltonian interactions $\mathcal{P}$ can generate another set $\mathcal{Q}$, we need to check the condition $\langle \mathcal{P} \rangle_{Lie} = \langle \mathcal{P} \cup \mathcal{Q} \rangle_{Lie}$. 

In the QAOA set-up, we have the control Hamiltonians $H_Z$ and $H_X$, and we are interested in knowing whether the Lie algebra $\mathcal{L} = \langle iH_Z, iH_X \rangle_{Lie}$ generates (up to a phase) the entire unitary group $\mathcal{U}(2^n)$.
In the examples to follow, we treat families of QAOA gates when universality holds and also mention cases when it doesn't.
Our main proof strategy will be to show either that  $e^{\mathcal{L}} $ contains some gates that are already known to form a universal gate set, or to show that due to some symmetry property we cannot reach all gates.

\section{Proving universality in 1-D set up}\label{sec:universality-1d}

In \cite{Lloyd-proof}, a  derivation was given for the universality of the QAOA in terms of two Hamiltonians defined on a $1-$dimensional system. Here we give the complete proof and the precise conditions under which such a QAOA is universal. 

We start by defining the Hamiltonians in a 1 dimensional line as in \cite{Lloyd-proof}

\begin{equation}
\begin{split}
\label{eq:Hz}
H_Z &= \sum_j \omega_A Z_{2j} + \omega_B Z_{2j+1} + \gamma_{AB} Z_{2j}Z_{2j+1} + \gamma_{BA} Z_{2j+1} Z_{2j+2} \\
    &= \omega_A H_A + \omega_B H_B + \gamma_{AB} H_{AB} + \gamma_{BA} H_{BA},
\end{split}
\end{equation}
\begin{equation}\label{eq:Hx}
    H_X = \sum_j X_j.
\end{equation}

We shall prove that when the number of qubits $n$ is odd then the QAOA defined with the previous Hamiltonians is universal. For the $n$ even case we will see this is not the case. A graph representing the Hamiltonian $H_Z$ for $n=6$ is shown in Fig. \ref{fig:1d-line}.

 \begin{figure}[htbp]
\centering
\includegraphics[scale=1.3]{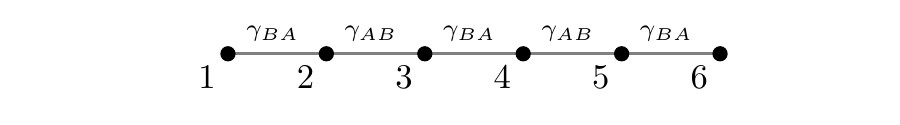}
\caption{System corresponding to Hamiltonian in Eq.~\eqref{eq:Hz} for $n=6$. Each node corresponds to qubits in the system and the edges to a two-body interaction. }
\label{fig:1d-line}
\end{figure}

For clarity, we make explicit the limits of the sums for each term in $H_Z$. Furthermore, we write in the upper limits of the sums the corresponding limits for {$n$ even $|$ $n$ odd}.
\begin{align}
    H_A &= \sum_{j=1}^{ \frac{n}{2} | \frac{n-1}{2}} Z_{2j}, & H_B &= \sum_{j=0}^{\frac{n}{2} -1 | \frac{n-1}{2}}  Z_{2j+1}, \\
     H_{AB} &= \sum_{j=1}^{\frac{n}{2} -1 | \frac{n-1}{2}} Z_{2j} Z_{2j+1}, & H_{BA} &= \sum_{j=0}^{\frac{n}{2} -1 | \frac{n-3}{2}} Z_{2j+1} Z_{2j+2},\\
     H_X &= \sum_{j=1}^{n} X_j.
\end{align}
It will also be useful to define
\begin{align}
     X_{odd} &=  \sum_{j=0}^{\frac{n}{2}-1 | \frac{n-1}{2}} X_{2j+1}, & X_{even}&= \sum_{j=1}^{\frac{n}{2} | \frac{n-1}{2}} X_{2j}.
\end{align}
We will start by proving the following lemma.

\begin{lemma}\label{lemma:decouple}
$iH_{Z1}=\omega_A iH_A + \omega_B iH_B \in \mathcal{L} =\langle \{iH_Z, iH_X \} \rangle_{Lie}$. Note that as a consequence we have that $iH_{Z2} = \gamma_{AB} iH_{AB} + \gamma_{BA} iH_{BA} \in \mathcal{L}$  
\end{lemma}
\begin{proof}
Consider first the commutator
\begin{equation}
\begin{split}
H_{YZ} = \frac{1}{2i}[H_Z,H_X]& = \omega_A \sum_{j=1}^{\frac{n}{2} | \frac{n-1}{2}} Y_{2j} + \omega_B  \sum_{j=1}^{\frac{n}{2}-1 | \frac{n-1}{2}} Y_{2j+1}  \\
 &  + \gamma_{AB} \sum_{j=1}^{\frac{n}{2}-1 | \frac{n-1}{2}} (Y_{2j}Z_{2j+1} + Z_{2j}Y_{2j+1})\\
 &+ \gamma_{BA} \sum_{j=0}^{\frac{n}{2}-1 | \frac{n-3}{2}}  (Y_{2j+1}Z_{2j+2} + Z_{2j+1}Y_{2j+2}), \\
\end{split}
\end{equation}
then, let us perform the calculation
\begin{equation}
    \begin{split}
        \frac{1}{2i}[H_{YZ}, H_X] &= -\omega_A \sum_{j=1}^{\frac{n}{2} | \frac{n-1}{2}} Z_{2j} - \omega_B  \sum_{j=1}^{\frac{n}{2}-1 | \frac{n-1}{2}} Z_{2j+1}\\
        & + \gamma_{AB} \sum_{j=1}^{\frac{n}{2}-1 | \frac{n-1}{2}} 2(Y_{2j} Y_{2j+1} - Z_{2j}Z_{2j+1})\\
        & + \gamma_{BA} \sum_{j=0}^{\frac{n}{2}-1 | \frac{n-3}{2}}  2(Y_{2j+1}Y_{2j+2} - Z_{2j+1}Z_{2j+2}), \\
    \end{split}
\end{equation}
and define
\begin{equation}
    \begin{split}
        H_{(1)} &= \frac{1}{2i}[H_{YZ},H_X] + H_Z\\
        &= 2\gamma_{AB} \sum_{j=1}^{\frac{n}{2}-1 | \frac{n-1}{2}} Y_{2j} Y_{2j+1} + 2\gamma_{BA} \sum_{j=0}^{\frac{n}{2}-1 | \frac{n-3}{2}}  Y_{2j+1}Y_{2j+2}\\
        & - \gamma_{AB} \sum_{j=1}^{\frac{n}{2}-1 | \frac{n-1}{2}}  Z_{2j}Z_{2j+1} - \gamma_{BA} \sum_{j=0}^{\frac{n}{2}-1 | \frac{n-3}{2}} Z_{2j+1}Z_{2j+2}.
    \end{split}
\end{equation}
Next, define also 
\begin{equation}
    \begin{split}
    H_{(2)} &= \frac{1}{2i}[H_{(1)},H_X] \\
            &= -3\gamma_{AB} \sum_{j=1}^{\frac{n}{2}-1 | \frac{n-1}{2}} (Y_{2j} Z_{2j+1} + Z_{2j} Y_{2j+1})\\
            & - 3\gamma_{BA} \sum_{j=0}^{\frac{n}{2}-1 | \frac{n-3}{2}} (Y_{2j+1} Z_{2j+2} + Z_{2j+1} Y_{2j+2}).
     \end{split}
\end{equation}
Finally, notice that we have
\begin{equation}
    \frac{1}{2i}[H_{YZ} + \frac{1}{3} H_{(2)} , H_X] = H_{Z2},
\end{equation}
which completes the proof. \qed
\end{proof}

Next, we prove that it is possible to generate $X_{even}$ and $X_{odd}$

\begin{proposition}\label{prop:even-odd}
Let $\omega_A^2 \neq \omega_B^2$, then $iX_{even}$, $iX_{odd}$ $\in \mathcal{L} = \langle iH_Z, iH_X \rangle_{Lie}$ 
\end{proposition}
\begin{proof}
From Lemma \ref{lemma:decouple}, we have that $iH_{Z1} = \omega_A iH_A + \omega_B iH_B$, $iH_{Z2} = \gamma_{AB} iH_{AB} + \gamma_{BA} iH_{BA}$ $\in \mathcal{L}$.

Next, let us define the following element in the Lie algebra
\begin{equation}
\begin{split}
    H_{Y1} &= \frac{1}{2i}[H_{Z1}, H_{X}] \\
        &= \omega_A \sum_{j=1}^{\frac{n}{2} | \frac{n-1}{2}}  Y_{2j} + \omega_B \sum_{j=0}^{\frac{n}{2}-1 | \frac{n-1}{2}} Y_{2j+1},
\end{split}
\end{equation}
and then calculate the commutator
\begin{equation}
    \begin{split}
        \frac{1}{2i}[H_{Z1},H_{Y1}] &=  \omega_A^2 \sum_{j=1}^{\frac{n}{2} | \frac{n-1}{2} } X_{2j} + \omega_B^{2} \sum_{j=0}^{\frac{n}{2}-1 | \frac{n-1}{2} } X_{2j+1}.
    \end{split}
\end{equation}
Now notice that
\begin{equation}
    \begin{split}
        \omega_A^2 H_X -  \omega_A^2 \sum_{j=1}^{\frac{n}{2} | \frac{n-1}{2} } X_{2j} - \omega_B^{2} \sum_{j=0}^{\frac{n}{2}-1 | \frac{n-1}{2} } X_{2j+1} = (\omega_A^2 - \omega_B^2)  \sum_{j=0}^{\frac{n}{2}-1 | \frac{n-1}{2} } X_{2j+1},
    \end{split}
\end{equation}
which implies that if $\omega_A^2 \neq \omega_B^2$, then  $iX_{even}$, $iX_{odd}$ 
$\in \mathcal{L}$. \qed
\end{proof}

From what we have so far proved, we can then generate $H_A, H_B, H_{AB}, H_{BA}$. The following proposition states the conditions for this.

\begin{restatable}{proposition}{sepgenerators}\label{prop:sep-generators}
 Assume $\gamma_{AB}^2 \neq \gamma_{BA}^2$ and let $\gamma = (\gamma_{AB}^2 - 4\gamma_{BA}^2)$.
 If  $\gamma \neq 0$, $\gamma_{AB}^2 \neq 0$, $\gamma_{BA}^2 \neq 0$, then $iH_A$, $iH_B$, $iH_{AB}$, $iH_{BA} \in \langle iH_Z, iH_X \rangle_{Lie}$.
\end{restatable}

The proof of Proposition~\ref{prop:sep-generators} is given in Appendix \ref{sec:appendix-proofs}. Note that in Ref.~\cite{Lloyd-proof} it was required that $\omega_A, \omega_B, \gamma_{AB}, \gamma_{BA}$ be rationally independent. In our  proof of universality this will be relaxed to the condition given by Proposition~\ref{prop:sep-generators}. 

In the following, we will prove that when $n$ is odd and the condition of the previous lemmas and propositions are fulfilled, then QAOA can implement all single qubit operators and $CNOT$.

\begin{restatable}{lemma}{separationX}\label{lemma:separation-X}
 Assume $n$ is odd, then   $iX_j \in \langle iH_A, iH_B, iH_{AB}, iH_{BA}, iH_X \rangle_{Lie}$  \\
 for any $j \in \{1, \ldots, n\}$.
\end{restatable}
The proof for Lemma \ref{lemma:separation-X} is given in Appendix \ref{sec:appendix-proofs}.

\begin{theorem}\label{thm:universality-1d}
Given an odd integer $n$, $H_Z$ as in Eq.~\eqref{eq:Hz}, $H_X$ as in Eq.~\eqref{eq:Hx}, with coefficients in $H_Z$ and $H_X$ fulfilling the conditions of Proposition \ref{prop:sep-generators} and $\mathcal{L} =\langle iH_Z, iH_X \rangle_{Lie}$, then $e^{\mathcal{L}}$ is dense in $\mathcal{U}(n)$. This implies universality for odd integers in QAOA.
\end{theorem}
\begin{proof}
We proved in Lemma \ref{lemma:separation-X} that $R_X(\theta) = e^{\frac{i}{2}X \theta } \in e^{\mathcal{L}}$ it is easy to see that also $R_Y(\phi), R_Z(\psi)  \in e^{\mathcal{L}}$. Thus, all single qubit operators are in $e^{\mathcal{L}}$. If it is possible to generate a two qubit gate such as $CNOT$, then we can prove that $\mathcal{L}$ can generate any unitary by, for example, generating the gate set of Clifford gates + $T$, which are known to be universal for quantum computation. In fact, any $2$-qubit entangling operator with all $1$-qubit gates is enough for universality \cite{universal-gates}.

In the proof of Lemma \ref{lemma:separation-X} we have not only managed to generate $1$-qubit Pauli's but also $2$-qubit Pauli's such as $Z_{k-1} Z_k$. To see that $CNOT$  gates can be generated, recall that $CNOT = \ketbra{0}{0} \otimes \eye + \ketbra{1}{1} \otimes X = \frac{1}{2} ( \eye \otimes \eye + Z \otimes \eye + \eye \otimes X - Z \otimes X ) $. Note that this last expression is in $\mathcal{L}$.

Finally, note that 

\begin{equation}
    \begin{split}
        e^{i\frac{\pi}{4} (\eye \otimes \eye - \eye \otimes X - Z \otimes \eye + Z \otimes X)} &= e^{i\frac{\pi}{4} (1-\eye \otimes X)(1 - Z \otimes \eye)}\\
         &= CNOT.
    \end{split}
\end{equation}

Since $\eye \otimes \eye - X_2 - Z_1 + Z_1 X_2$ is in $\mathcal{L}$, we conclude that $CNOT$ can be generated. \qed
\end{proof}

With this we have proved universality for $n$ odd. It is easy to see that for $n$ even $\langle iH_Z, iH_X \rangle_{Lie}$ cannot approximate $\mathcal{U}(2^n)$  due to the presence of a symmetry in the system. This is easier to see with a concrete example, if $n=4$ and we number qubits from $1$ to $4$ then exchanging qubit $1$ with qubit $4$ and exchanging qubit $2$ with qubit $3$ is a symmetry of the system. The presence of a symmetry in Hamiltonians $H_Z$ and $H_X$ imply non-universality; let $U$ be the unitary implementing the symmetry commuting with both Hamiltonians, then $H_Z$ and $H_X$ can be block diagonalized which necessarily implies that there are elements in $\mathcal{U}(2^n)$ that can't be approximated.

\section{Universality for QAOA defined on graphs}\label{sec:universality-general}

In Section \ref{sec:universality-1d}, we proved universality in a particular setting of a QAOA. Here we show that universality can be obtained also in more general settings. The algorithms defined here are characterized by the choice of the Hamiltonians $H_Z$ and $H_X$. To define $H_Z$, we make a correspondence between a non-directed simple graph (no loops or multiple edges) $G=(V,E)$ and the terms appearing in $H_Z$, while the Hamiltonian $H_X$ is defined as in Section \ref{sec:universality-1d}.

\subsection{Universality from zero forcing}
We prove in this section that the property of universality on this class of QAOA is present depending on a process defined on the graph $G$ called zero forcing. The notion of zero forcing has been presented before in the context of quantum control on graphs \cite{Burgarth-control,zeroforce} and we find that it applies as well in this context.

\begin{definition}[Zero forcing]
Consider a simple graph $G=(V,E)$, a zero forcing process on $G$ consists of an initial set of vertices $S \subseteq V$ which we will consider as ``infected''. The rest of the vertices are non infected. Then we proceed by steps to infect other nodes, at each step an infected vertex $v$  infects a non infected neighbour $w$ if $w$ is the only non infected neighbour of $v$. We call $S$ a zero forcing set if we can infect all the graph by starting with all infected vertices in $S$.
\end{definition}

As usual with QAOA, we start defining two Hamiltonians $H_Z$ and $H_X$. Consider simple graph $G=(V,E)$ and a subset $S \subseteq V$. 

\begin{equation}\label{eq:H_Z-0force}
    \begin{split}
        H_Z &= \gamma \sum_{(i,j) \in E} Z_i Z_j + \sum_{i\in S} \omega_i Z_i + \omega \sum_{i \in V \setminus S} Z_i\\
        &= \gamma H_{\gamma} + \sum_{i\in S} \omega_i Z_i + \omega H_V,
    \end{split}
\end{equation}

\begin{equation}\label{eq:H_X-0force}
    \begin{split}
        H_X = \sum_{i\in V} X_i.
    \end{split}
\end{equation}

\begin{theorem}\label{thm:zero-force}
Let $G=(V,E)$ be a simple graph and $S\subseteq V$. Define $H_Z$ and $H_X$ as in Eq.~\ref{eq:H_Z-0force} and Eq.~\ref{eq:H_X-0force} and let  $\gamma$, $\omega_i$, $\omega$ be rationally independent. Consider $S$ as the inital set of infected nodes in a zero forcing process. If $S$ is a zero forcing set, then $Z_k Z_j \in \langle H_Z, H_X \rangle_{Lie}$ for all $(k,j) \in E$ and $X_k \in \langle H_Z, H_X \rangle_{Lie}$ for all $k \in V$.
\end{theorem}

\begin{proof}
Since  $\gamma$, $\omega_i$, $\omega$ are rationally independent, using a similar method to the proof in Proposition \ref{prop:separating-AB-BA} (see Appendix \ref{sec:appendix-separate}) we can generate $H_\gamma, H_V, Z_i$ for $i \in S$
First note that for vertices $i \in S$ we can generate $X_i$. Consider two vertices $i,j \in S$ such that they are neighbouring vertices in $G$. To see this, commute

\begin{equation}
    \frac{1}{(2i)^2}[[H_\gamma,X_i],X_j]= Y_i Y_j.
\end{equation}
Thus, we can also generate $Z_i Z_j$.
Consider now $i\in S$ that only has one neighbour $j\in V\setminus S$. We show that we can generate $X_j$. Define $H_i$ as $H_\gamma$ with the interaction terms corresponding to infected neighbours of $i$ subtracted. Consider now the commutator:

\begin{equation}
\frac{1}{2i}[X_i,H_i] = Y_i Z_j .   
\end{equation}
And thus $Z_i Z_j$ can be generated. Then we can commute with $H_X - X_i$ and generate $Z_i Y_j$ which commuted with $Z_i Z_j$ generates $X_j$. This is analogous to an infection step in the zero forcing process. It is then easily seen that if $S$ is zero forcing, then all single qubit and two qubit operators are generated in the graph.
\end{proof}

We can generalize even more this zero forcing process by difference considering edge interactions in $H_Z$. Given once again a graph $G=(V,E)$ and set $S\subseteq V$, consider now that we can partition the set of edges $E$ into $q$ disjoint sets  $\{E_i\}_{i\in [q]}$ such that $\bigcup_{i\in [q]} E_i = E$. From this we write the Hamiltonian

\begin{equation}\label{eq:H_Z-0force-mod}
    \begin{split}
        H_Z &= \sum_{k=1}^q  \sum_{(i,j) \in E_k} \gamma_k Z_i Z_j + \sum_{i\in S} \omega_i Z_i + \omega \sum_{i \in V \setminus S} Z_i\\
        &=\sum_{k=1}^q \gamma_k H_{\gamma_k} + \sum_{i\in S} \omega_i Z_i + \omega H_V.
    \end{split}
\end{equation}

\begin{definition}[Generalized zero forcing for multi-type edges]
Consider a simple graph $G=(V,E)$ with $E=\bigsqcup_{i\in [q]} E_i$, a zero forcing process on $G$ consists of an initial set of vertices $S \subseteq V$ which we will consider as ``infected''. The rest of the vertices are non infected.

The generalized zero forcing process proceeds in one step by considering each infected vertex and the subgraph $G_1 =(V,E_1)$. If an infected vertex has a single non infected vertex in $G_1$, then infect this new vertex and add it to $S$. Then proceed in the same fashion with the neighbours of vertices on $S$ in graphs $G_2,G_3,..,G_q$. Repeat this process and if the whole graph ends infected then we call the initial set $S$ a generalized zero forcing set. 
\end{definition}
We prove the following result.
\begin{theorem}
Let $G=(V,E)$ be a simple graph, $S\subseteq V$ and consider a  partition of the set of edges $E$ into $q$ disjoint sets  $\{E_i\}_{i\in [q]}$ such that $\bigcup_{i\in [q]} E_i = E$. Define $H_Z$ and $H_X$ as in Eq.~\eqref{eq:H_Z-0force-mod} and Eq.~\eqref{eq:H_X-0force} and let  $\gamma$, $\omega_i$, $\omega$ be rationally independent. Consider $S$ as the inital set of infected nodes in a zero forcing process. If $S$ is a generalized zero forcing set, then $\langle H_Z, H_X \rangle$ generates $Z_k Z_j$ for all $(k,j) \in E$ and $X_k$ for all $k \in V$.
\end{theorem}

\begin{proof}
The proof is almost the same as in Theorem \ref{thm:zero-force}. \qed
\end{proof}

\subsection{Universality without zero forcing}\label{sec:universality-grid}
Note that a Hamiltonian defined from a graph and a initial subset of vertices $S$ may not have a zero forcing set, yet nonetheless can be universal. We will give one such an example with a two dimensional grid with only two edges under control. This example points to a more general process than zero forcing that allows to study universality in the corresponding QAOA, although we will not pursue this direction in this work.

Define a graph composed of a square grid with $N = n^2$ vertices, number the vertices from $v_1$ to $v_N$ left to right and top to bottom . We assume all interactions in the grid are labeled by the same interaction type $A$. We also add two extra nodes labeled $v_{N+1}$ and $v_{N+2}$. Connect $v_{N+1}$ to vertex $v_{1}$ with an edge labeled $B$ and connect $v_{N+2}$ to $v_{N}$ with an edge labeled $C$. We give an example for $N=25$ in Fig. \ref{fig:graph-grid}.

For this graph we define the following Hamiltonians:

\begin{equation}\label{eq:H_Z-grid}
\begin{split}
    H_Z &= \omega_A \sum_{v_i \in V_{Grid}} Z_{v_i} + \omega_B  Z_{v_{N+1}} + \omega_C  Z_{v_{N+2}}\\
    &+  \gamma_{A} \sum_{(v_i,v_j)\in E_{Grid}} Z_{v_i} Z_{v_j} + \gamma_{B}  Z_{v_{1}} Z_{v_{N+1}} + \gamma_{C}  Z_{v_{n}} Z_{v_{N+2}},
\end{split}
\end{equation}

\begin{equation}\label{eq:H_X-grid}
   H_X = \sum_{i=1}^{N+2} X_i.
\end{equation}

 \begin{figure}[htbp]
\centering
\includegraphics[]{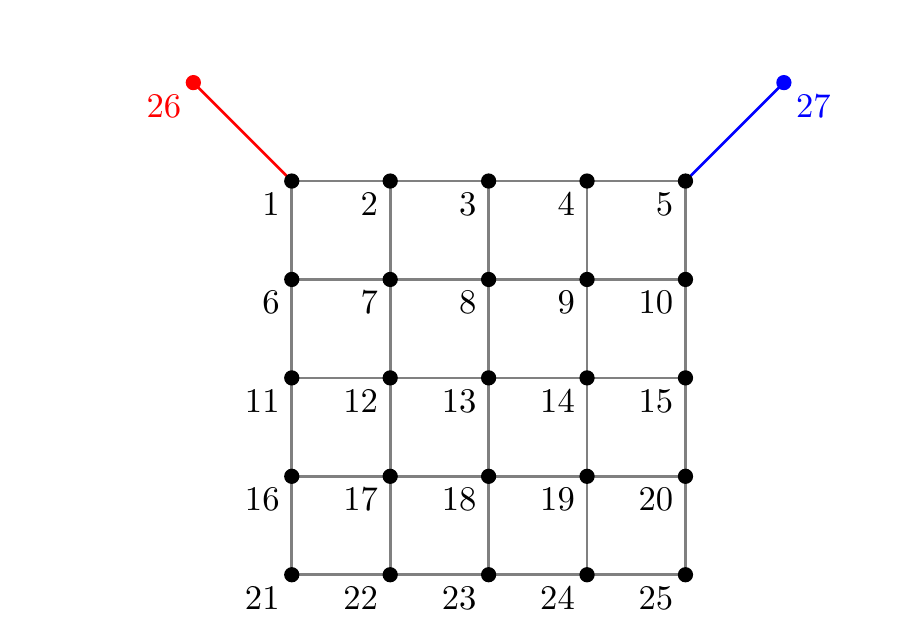}
\caption{Grid with $N=25$ nodes which defines a Hamiltonians as in Eq.~\eqref{eq:H_Z-grid}. Vertices $26$ and $27$ correspond to qubits where $H_Z$ acts with one-qubit operators with coefficients $\omega_B$ and $\omega_C$, the corresponding incident edges define two-qubit interactions with coefficients $\gamma_B$ and $\gamma_C$. (Color online)}
\label{fig:graph-grid}
\end{figure}

We want to prove that every single qubit operator $Z_i$ and two body operators $Z_i Z_j$ can be generated. 

Note that Lemma \ref{lemma:decouple} applies in this situation as well, so we can separate 

$$H_{Z1} = \omega_A H_{A1} + \omega_B Z_{N+1} + \omega_C Z_{N+1},  $$
$$H_{Z2} = \gamma_A H_{A2} + \gamma_B  Z_{v_{1}} Z_{v_{N+1}} + \gamma_C Z_{v_{n}} Z_{v_{N+2}}. $$

From this, we easily see that we can generate as well $X_{N+1}$, $X_{N+2}$ and $X_{Grid} = \sum_{i=1}^{N} X_i$. Finally notice that generating $H_{A1}$, $Z_{N+1}$, $Z_{N+1}$, $H_{A2}$, $Z_{v_{1}} Z_{v_{N+1}}$, $Z_{v_{n}} Z_{v_{N+2}} $ separately can be done applying Proposition \ref{prop:separating-AB-BA}.

To prove that any gate can be generated with these Hamiltonians, we prove that all $Z_j$ with $j\in\{1,..,n\}$ and $Z_k Z_{k+1}$ with $k \in \{1,..,n-1\}$ can be generated. In this way there is full controllability of the first horizontal line in the grid. After proving this, it directly follows that QAOA defined from the grid is universal by the zero forcing argument.

\begin{restatable}{theorem}{grid}
  Given a graph $G$ as described above, vertices numbered in the order mentioned previously, and given the Hamiltonians $H_{A1}$, $Z_{N+1}$, $Z_{N+2}$, $X_{N+1}$, $X_{N+2}$, $X_{Grid}$, $H_{A2}$, $Z_{1} Z_{N+1}$, $Z_{n} Z_{N+2}$, then for $i\in \{1,..,n-1\}$ and $j\in \{1,..,n\}$ we have that $Z_j$, $X_j$ ,$Z_i Z_{i+1} \in \langle H_{A1}$, $Z_{N+1}$, $Z_{N+1}$, $H_{A2}$, $Z_{1} Z_{N+1}$, $Z_{n} Z_{N+2}\rangle_{Lie}$. This implies universality for any $n$ on the grid.
\end{restatable}

The proof is given in Appendix \ref{sec:appendix-grid}. As mentioned before this points to a more general process that allows to show universality but for brevity we won't go further in this direction.

\section{Universality for QAOA defined on hypergraphs}\label{sec:hypergraphs}
So far the Hamiltonians $H_Z$ induced by graphs define only quadratic or linear terms. We can consider higher order terms for $H_Z$ by studying a modified version of a zero forcing process on hypergraphs. Here we will consider the specific case of Hamiltonians with cubic terms as there is already work studying problems with cubic order term Hamiltonians as in the MAXE3LIN2 problem \cite{maxe3lin2}.

 From a hypergraph we can define Hamiltonians $H_Z$ with $k-$body terms where $k>2$. A hypergraph $\mathcal{G}=(V,E)$ is a generalization of a graph, it is defined by a finite set of vertices $V$ and a finite set $E$ which contains non empty subsets of $V$ which are called hyperedges. In Fig. \ref{fig:hypergraph-line} we show an example of a hypergraph defined by $V=\{v_1,v_2,...,v_6\}$ and
 
 $$E=\bigg\{ \{v_1,v_2,v_3 \}, \{v_2, v_3, v_4 \}, \{v_3, v_4, v_5 \}, \{v_4, v_5, v_6 \} \bigg\}$$
 
 This is also an example of a $3-$uniform hypergraph, a $k-$uniform hypergraph is one where all hyperedges have exactly $k$ nodes.

 \begin{figure}[htbp]
\centering
\includegraphics[]{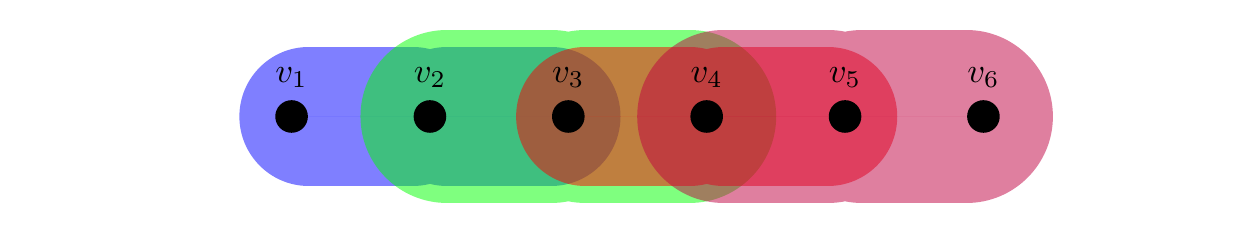}
\caption{Example of a $3-$uniform hypergraph on a line with every hyperedge contains three edges.(Color online)}
\label{fig:hypergraph-line}
\end{figure}

We will prove here universality on $3-$uniform hypergraphs with a small modification in the Hamiltonian defined from the hypergraph. Consider a hypergraph $\mathcal{G}=(V,E)$ with $V=\{1,..,n\}$ and $E=\bigg\{\{1,2 \}, \{1,2,3\},\{2,3,4\},...,\{n-2,n-1,n\} \bigg\}$. An example for $n=6$ is shown in Fig. \ref{fig:hypergraph-line} (without the $2-$edge). 

From $\mathcal{G}$ we define the following Hamiltonians

\begin{equation}\label{eq:H_Z-hyperline}
    \begin{split}
        H_Z &= \delta \sum_{\{i,j,k\} \in E} Z_i Z_j Z_k + \gamma Z_1 Z_2 +  \omega_1 Z_1 + \omega \sum_{i \neq 1} Z_i\\
        &= \delta H_{\delta} + \gamma Z_1 Z_2 +  \omega_1 Z_1 + \omega H_V,
    \end{split}
\end{equation}

\begin{equation}\label{eq:H_X-hyperline}
    \begin{split}
        H_X = \sum_{i\in V} X_i.
    \end{split}
\end{equation}

We wish to generate all $2-$qubit operators between neighbours and $1-$qubit operators on every vertex. This hyper-zero forcing is defined by starting with some initial set of infected vertices $S_1$ and a set of infected $2-$edges $S_2$; at each step pick an infected vertex, if it has only one non infected $3-$neighbour then infect the neighbour. If two infected $3-$neighbours share a a $2-$edge and then connect each infected node to the non infected one with $2-$edges. 

In the $3-$uniform hypergraph, the infection step in terms of the commutators proceeds as follows, first note that the term $Z_1 Z_2 Z_3$ can be separated from the other cubic terms and that $X_2$ can be easily separated, now consider

\begin{equation}
    \frac{1}{2i}[Z_1 Y_2, Z_1 Z_2 Z_3] = X_2 Z_3.
\end{equation}

From this we see that $X_3$ can be separated and we can proceed to separate $Z_2 Z_3 Z_4$. In this way we proceed until the end of the chain having produced all one qubit and two qubit operators between neighbours which proves universality.

We can define a hyper-zero forcing procedure on hypergraphs which allows to check if the corresponding QAOA is universal. We will write here for conciseness only the case of hypergraphs with hyperedges with at most three elements although a more generalized version is possible

\begin{definition}
Consider a hypergraph $\mathcal{G}=(V,E)$ where $|e|\leq 3$ for all $e \in E$, a hyper-zero forcing process on $\mathcal{G}$ consists of an initial set of vertices $S_1 \subseteq V$ and an initial set of $2-$edges $S_2$ which we will consider as 'infected'. The rest of the vertices and $2-$edges are non infected. Then we proceed by steps to infect other nodes, at each step a pair of infected vertices $v_1, v_2$  infects a non infected $3-$neighbour $w$ if $w$ is the only non infected $3-$neighbour of $v_1$ and $v_2$ and also the $2-$edge ${v_1, v_2}$ is infected.  We call $S_1$ and $S_2$ hyper-zero forcing sets if we can infect all the graph by starting with $S_1$ and $S_2$ infected.
\end{definition}

An analogous theorem can be derived as in the zero forcing case for relating hyper-zero forcing processes and universality. Here, for simplicity, we state such theorem for hypergraphs with hyperedges containing three or less vertices.

\begin{theorem}
Let $\mathcal{G}=(V,E)$ be a hypergraph with $|e| \leq$, $S_1 \subseteq V$ and $S_2$ a set of $2-$edges. Define $H_Z$ and $H_X$ as in Eq.~\ref{eq:H_Z-hyperline} and Eq.~\ref{eq:H_X-hyperline} and let  all coefficients in $H_Z$ be rationally independent. Consider $S_1$ as the initial set of infected nodes and $S_2$ as the set of infected edges in a hyper-zero forcing process. If $S_1$ and $S_2$ are  hyper-zero forcing sets, then $Z_k Z_j \in \langle H_Z, H_X \rangle_{Lie}$ for all $(k,j) \in E$ and $X_k \in \langle H_Z, H_X \rangle_{Lie}$ for all $k \in V$.
\end{theorem}

\begin{proof}
Proof follows directly from arguments in the $3-$uniform hyper graph case and similarly as in the zero process case. \qed
\end{proof}

In a previous work \cite{bremner2004}, it was shown that local unitaries and unitaries generated by three-body Pauli-operators does not give rise to universality. This directly implies the following no-go result:

\begin{theorem}
Define $H_Z$ and $H_X$ as in Eq.~\ref{eq:H_Z-hyperline} and Eq.~\ref{eq:H_X-hyperline}. If the coefficient $\gamma$ in $H_Z$ is zero, then the QAOA defined by $H_Z$ and $H_X$ does not yield a universal gate set.
\end{theorem}

\section{Conclusion and outlook}\label{sec:conclusion}

We proved the computational universality of different QAOA set-ups. In particular, we completed an earlier proof for a specific set-up given in Ref.~\cite{Lloyd-proof}, and also found two new broad classes of driver Hamiltonians that allow the corresponding QAOA unitaries to perform universal quantum computation. The first class consists of Hamiltonians with quadratic and linear terms;  the quadratic terms are distributed according to the adjacency matrix of a graph, while the coupling strength of the linear terms are grouped into two parts defined by a so-called zero forcing set of the graph. This construction was then generalized to obtain a second class of driver Hamiltonians with higher order terms corresponding to hypergraphs and generalized zero forcing sets. Here it should also be mentioned that the square grid example, presented in Sect.~\ref{sec:universality-grid}, points to a more general graph process different from zero forcing that may further advance an understanding of universality in QAOA circuits
(and perhaps also in more general quantum control set-ups). Another important generalization of our results would be to regard other mixer Hamiltonians then the type $H_X = \sum_i X_i$ considered here, e.g., one could consider $XY$ mixers \cite{QAOA-hard-soft,XY-mixers}. One could hope to determine more general conditions for universality of QAOA unitaries, which could include the above mentioned generalizations; we leave this for future work. Such general results could help in understanding the relation between the choice of Hamiltonians and the space reached by the ansatz in the algorithm, and perhaps also to obtain some analytical results about the efficiency of QAOA. We regard our work as a first step towards this goal.

\begin{acknowledgements}
We would like to thank discussions with Micha{\l} Oszmaniec. ZZ was supported by the Hungarian National Research, Development and In- novation Office (NKFIH) through Grants No. K124351, K124152, K124176 KH129601, and the Hungarian Quantum Technology National Excellence Program (Project No. 2017-1.2.1-NKP-2017- 00001); and he was also partially funded by the János Bolyai and the Bolyai+ Scholarships.
\end{acknowledgements}

\bibliographystyle{spphys}       

\bibliography{references}

\newpage 
\appendix
\noindent {\bf {\large{APPENDIX}}}

\section{Proof of some results in Section \ref{sec:universality-1d}} \label{sec:appendix-proofs}

\sepgenerators*
\begin{proof}
From Proposition \ref{prop:even-odd}, we see that $H_A$ and $H_B$ can be easily generated.
To prove that $H_{AB}$ and $H_{BA}$ can be generated, we separate the proof for $n$ odd and $n$ even case. \\

\underline{$n$ odd}:
\begin{equation}
\begin{split}
[H_{Z2}, X_{even}] &= \gamma_{AB} \sum_{j=1}^{\frac{n-1}{2}} Y_{2j} Z_{2j+1} +                        \gamma_{BA} \sum_{j=0}^{\frac{n-3}{2}} Z_{2j+1} Y_{2j+2}\\
                    &= H_{YZ}^{e}.
\end{split}
\end{equation}

Then

\begin{equation}
    \begin{split}
        [H_{YZ}^{e} , H_{Z2}] &= \gamma_{AB}^2  \sum_{j=1}^{\frac{n-1}{2}} X_{2j} + 2\gamma_{AB} \gamma_{BA}  \sum_{j=1}^{\frac{n-1}{2}} Z_{2j-1}X_{2j}Z_{2j+1}\\
        & \hspace{5mm} + \gamma_{BA}^2  \sum_{k=0}^{\frac{n-3}{2}} X_{2j+2}.
    \end{split}
\end{equation}
Note that we have suppressed the $(2i)$ that appear from the commutators. The $\gamma_{AB}^2$ and $\gamma_{BA}^2$ terms in the last line can be removed, so we define

$$H_{ZXZ} =    \sum_{j=1}^{\frac{n-1}{2}} Z_{2j-1}X_{2j}Z_{2j+1}.$$

Consider now

\begin{equation}
    \begin{split}
        [H_{YZ}^{e}, H_{ZXZ}] &=  \gamma_{AB} \sum_{j=1}^{\frac{n-1}{2}} Z_{2j-1}X_{2j} + \gamma_{BA}  \sum_{j=1}^{\frac{n-1}{2}}  X_{2j} Z_{2j+1},\\
    \end{split}
\end{equation}
and define
$$H_{Z',2}^{odd} =   \gamma_{AB} \sum_{j=1}^{\frac{n-1}{2}} Z_{2j-1}X_{2j} + \gamma_{BA}  \sum_{j=1}^{\frac{n-1}{2}}  X_{2j} Z_{2j+1}.$$
Notice that
\begin{equation}
    \begin{split}
        \gamma_{AB} H_{Z',2}^{odd} - \gamma_{BA} H_{Z2} = (\gamma_{AB}^2 - \gamma_{BA}^2) \sum_{j=0}^{\frac{n-3}{2}} Z_{2j+1}Z_{2j+2}.
    \end{split}
\end{equation}
Thus, assuming $\gamma_{AB}^2 \neq \gamma_{BA}^2$ then we have generated $H_{AB}$ and $H_{BA}$ for odd $n$.\\

\underline{$n$ even}:

Following steps analogous to the odd $n$ case, we obtain

\begin{equation}
    \begin{split}
        [H_{Z2}, X_{odd}] &=\gamma_{AB} \sum_{j=1}^{\frac{n}{2}-1}  Z_{2j}     Y_{2j+1} + \gamma_{BA}  \sum_{j=0}^{\frac{n}{2}-1}  Y_{2j+1} Z_{2j+2}\\
                         &= H_{YZ}^{oo},
    \end{split}
\end{equation}

\begin{equation}
    \begin{split}
        [H_{YZ}^{oo}, H_{Z2}] &=\gamma_{AB}^2 \sum_{j=1}^{\frac{n}{2}-1} X_{2k+1}     + 2\gamma_{AB} \gamma_{BA} \sum_{j=1}^{\frac{n}{2}-1} Z_{2j}              X_{2j+1}Z_{2j+2} \\
                & \hspace{5mm} + \gamma_{BA}^2 \sum_{j=0}^{\frac{n}{2}-1} X_{2k+1}.
    \end{split}
\end{equation}

The last line is true up to a $(2i)$ factor. In the last line we can also eliminate the $\gamma_{BA}^2$ and define

$$H_{ZZZ1} = -\gamma_{AB} Z_1 + 2\gamma_{BA} \sum_{j=1}^{\frac{n}{2}-1} Z_{2j}Z_{2j+1}Z_{2j+2}  $$

\begin{equation}
    \begin{split}
        [H_{YZ}^{oo},H_{ZZZ1}] &=  2\gamma_{AB}  \sum_{j=1}^{\frac{n}{2}-1} Z_{2j+1} Z_{2j+2} -\gamma_{AB} Z_1 Z_2 + 2\gamma_{BA} \sum_{j=1}^{\frac{n}{2}-1} Z_{2j} Z_{2j+1}= H_{12}
    \end{split}
\end{equation}

Now we perform a similar calculation but using $X_{even}$.

\begin{equation}
    \begin{split}
        [H_{Z2},X_{even}] &= \gamma_{AB} \sum_{j=1}^{\frac{n}{2}-1} Y_{2j} Z_{2j+1} +\gamma_{BA} \sum_{j=0}^{\frac{n}{2}-1} Z_{2j+1} Y_{2j+2}\\
        &= H_{YZ}^{ee}
    \end{split}
\end{equation}

\begin{equation}
    \begin{split}
        [H_{YZ}^{ee},H_{Z2}] &= \gamma_{AB}^2 \sum_{j=1}^{\frac{n}{2}-1} X_{2j} + 2\gamma_{AB} \gamma_{BA} \sum_{j=1}^{\frac{n}{2}-1} Z_{2j-1}X_{2j}Z_{2j+1} + \gamma_{BA}^2 \sum_{j=0}^{\frac{n}{2}-1} X_{2j+2}
    \end{split}
\end{equation}

We can remove the $\gamma_{BA}^2$ and define

$$H_{ZZZn} = -\gamma_{AB}Z_{n} + 2\gamma_{BA} \sum_{j=1}^{\frac{n}{2}-1} Z_{2j-1}X_{2j} Z_{2j+1} $$

\begin{equation}
    \begin{split}
        [H_{YZ},H_{ZZZn}]=2\gamma_{AB}\gamma_{BA} \sum_{j=1}^{\frac{n}{2}-1} Z_{2j-1}X_{2j} - \gamma_{BA} \gamma_{AB} Z_{n-1}X_{n} +2 \gamma_{BA}^2 \sum_{j=1}^{\frac{n}{2}-1} X_{2j} Z_{2j+1}
    \end{split}
\end{equation}
Thus, we define

$$H_{n-1,n} =  2\gamma_{AB} \sum_{j=1}^{\frac{n}{2}-1} Z_{2j-1}Z_{2j} -  \gamma_{AB} Z_{n-1}X_{n} +2 \gamma_{BA} \sum_{j=1}^{\frac{n}{2}-1} Z_{2j} Z_{2j+1}$$

Then we can generate

        \begin{equation}
            \begin{split}
                H_{(2)} &= H_{12} + H_{n-1,n}\\
        &= \gamma_{AB}(Z_1 Z_2 + Z_{n-1} Z_{n}) + 4\gamma_{AB}\sum_{j=1}^{\frac{n}{2}-2}Z_{2j+1}Z_{2j+2} + 4\gamma_{BA} \sum_{j=1}^{\frac{n}{2}-1} Z_{2j}Z_{2j+1}
            \end{split}
        \end{equation}

Now we generate

\begin{equation}
    \begin{split}
        \gamma_{AB} H_{(2)} - 4\gamma_{BA}H_{Z2} &= (\gamma_{AB}^2 - 4\gamma_{BA}^2)(Z_1 Z_2 + Z_{n-1} Z_n)\\
                            &\hspace{5mm} + (4\gamma_{AB}^2 - 4\gamma_{BA}^2) \sum_{j=1}^{\frac{n}{2}-2} Z_{2j+1} Z_{2j+2}\\
                            &= (\gamma_{AB}^2 - 4\gamma_{BA}^2)\sum_{j=0}^{\frac{n}{2}-1} Z_{2j+1}Z_{2j+2} + 3 \gamma_{AB}^2 \sum_{j=1}^{\frac{n}{2}-2} Z_{2j+1}Z_{2j+2}
    \end{split}
\end{equation}

Define $\gamma = (\gamma_{AB}^2 - 4\gamma_{BA}^2)$ and

\begin{equation}
    \begin{split}
        H_{(3)} = \sum_{j=0}^{\frac{n}{2}-1} Z_{2j+1}Z_{2j+2} + 3 \frac{\gamma_{AB}^2}{\gamma}  \sum_{j=1}^{\frac{n}{2}-2} Z_{2j+1}Z_{2j+2}
    \end{split}
\end{equation}

\begin{equation}
    \begin{split}
        H_{Z2} -\gamma_{AB}H_{(3)} &= \gamma_{AB} \sum_{j=1}^{\frac{n}{2}-1} Z_{2j} Z_{2j+1} - 3 \frac{\gamma_{AB}^2}{\gamma} \gamma_{BA}  \sum_{j=1}^{\frac{n}{2}-2} Z_{2j+1}Z_{2j+2}
    \end{split}
\end{equation}

Define $\Tilde{\gamma}_2 =  3 \frac{\gamma_{AB}}{\gamma} \gamma_{BA} $  and 

$$H_{*} =  \sum_{j=1}^{\frac{n}{2}-1} Z_{2j} Z_{2j+1} - \Tilde{\gamma}_2 \sum_{j=1}^{\frac{n}{2}-2} Z_{2j+1}Z_{2j+2} $$

On the other hand consider

\begin{equation}
    \begin{split}
        H_{(2)}-\gamma_{AB}H_{(3)} = 4\gamma_{BA}  \sum_{j=1}^{\frac{n}{2}-1} Z_{2j} Z_{2j+1} + (3\gamma_{AB} - 3\frac{\gamma_{AB}^3}{\gamma}) \sum_{j=1}^{\frac{n}{2}-2} Z_{2j+1}Z_{2j+2}
    \end{split}
\end{equation}

And define $\Tilde{\gamma}_1 = 3\gamma_{AB}(1-\frac{\gamma_{AB}^2}{\gamma})\frac{1}{4\gamma_{BA}}$

$$H_{\Box} =   \sum_{j=1}^{\frac{n}{2}-1} Z_{2j}Z_{2j+1} + \Tilde{\gamma}_1  \sum_{j=1}^{\frac{n}{2}-2} Z_{2j+1}Z_{2j+2} $$

Then

\begin{equation}
    \begin{split}
        H_{\circ} = H_{\Box} - H_{*} = (\Tilde{\gamma}_1 - \Tilde{\gamma}_2) \sum_{j=1}^{\frac{n}{2}-1} Z_{2j+1}Z_{2j+2}
    \end{split}
\end{equation}
Finally
\begin{equation}
    \begin{split}
        H_{\Box} - \frac{\Tilde{\gamma}_1}{(\Tilde{\gamma}_1 - \Tilde{\gamma}_2)}H_{\circ} =  \sum_{j=1}^{\frac{n}{2}-1} Z_{2j}Z_{2j+1} = \frac{1}{\gamma_{AB}} H_{AB}
    \end{split}
\end{equation}

Note that we have $\gamma \neq 0$, $\Tilde{\gamma}_1 \neq 0$, $\Tilde{\gamma}_2 \neq 0$, $\Tilde{\gamma}_1 \neq \Tilde{\gamma}_2$ and since  $\gamma \neq 0$, $\gamma_{AB}^2 \neq 0$, $\gamma_{BA}^2 \neq 0$, we can generate $H_{AB}$ and $H_{BA}$ which gives the result. \qed
\end{proof}

\separationX*
\begin{proof}
Let us first see that $iX_1 \in \mathcal{L}$. Consider 

\begin{equation}
    \begin{split}
    [H_{AB}, H_X] &= (2i) \sum_{j=1}^{\frac{n}{2} -1 | \frac{n-1}{2}} Z_{2j} Y_{2j+1} + Y_{2j} Z_{2j+1} 
    \end{split}
\end{equation}

Define $H_{YZ|AB} = \sum_{j=1}^{\frac{n}{2} -1 | \frac{n-1}{2}} Z_{2j} Y_{2j+1} + Y_{2j} Z_{2j+1}$ and consider
\begin{equation}
    \begin{split}
        [H_{YZ|AB}, H_{AB}] &= (2i)\sum_{j=1}^{\frac{n}{2} -1 | \frac{n-1}{2}} X_{2j+1} + (2i) \sum_{j=1}^{\frac{n}{2}  | \frac{n-1}{2}} X_{2j}
    \end{split}
\end{equation}

Notice that in the last sum, all $X$ Pauli matrices appear, except the one acting on qubit $1$. Thus,

$$H_X - \sum_{j=1}^{\frac{n}{2} -1 | \frac{n-1}{2}} X_{2j+1} + X_{2j} = X_1 $$

And we have that $X_1 \in \mathcal{L}$. Assume now that we want to generate $X_{k}$ and that we have generated $X_{k-1}$. If $k$ is even, then
\begin{equation}\label{eq:BA-X_k-1}
    \begin{split}
        [H_{BA}, X_{k-1}] &=  (2i)Y_{k-1}Z_k
    \end{split}
\end{equation}

\begin{equation}\label{eq:YZ-X_k-1}
    \begin{split}
        [Y_{k-1}Z_k, X_{k-1}] = (-2i) Z_{k-1}Z_k
    \end{split}
\end{equation}

And finally,

\begin{equation}
    \begin{split}
        [H_{XA}, Z_{k-1}Z_k] &= \sum_{j=1}^{\frac{n}{2} -1 | \frac{n-1}{2}} [X_{2j}, Z_{k-1}Z_k] = (-2i) Z_{k-1}Y_{k}
    \end{split}
\end{equation}

\begin{equation}
    \begin{split}
        [H_{BA}, Z_{k-1}Y_{k}] &= [\sum_{j=0}^{\frac{n}{2} -1 | \frac{n-3}{2}} Z_{2j+1} Z_{2j+2}, Z_{k-1}Y_{k} ] = (-2i) X_k
    \end{split}
\end{equation}

Now if $k$ is odd, 

\begin{equation}
    \begin{split}
        [H_{AB}, X_{k-1}] &= \sum_{j=1}^{\frac{n}{2} -1 | \frac{n-1}{2}} [Z_{2j}Z_{2j+1}, X_{k-1}]\\
                          &= (2i)Y_{k-1}Z_k
    \end{split}
\end{equation}

\begin{equation}
    \begin{split}
        [Y_{k-1}Z_k, X_{k-1}] = (-2i) Z_{k-1}Z_k
    \end{split}
\end{equation}

\begin{equation}
    \begin{split}
        [H_{XB}, Z_{k-1}Z_k] &= (-2i) Z_{k-1}Y_{k}
    \end{split}
\end{equation}

\begin{equation}
    \begin{split}
        [H_{AB}, Z_{k-1}Y_k] = (-2i) X_k
    \end{split}
\end{equation}
Which proves the result. \qed
\end{proof}

\section{Proofs of results in Section \ref{sec:universality-general}}\label{sec:appendix-separate}
The following proposition shows that commuting terms in a Hamiltonian can be separated. 
\begin{proposition}\label{prop:separating-AB-BA}
 Let $H_{Z1} = \omega_A H_A + \omega_B H_B$ and $H_{Z2} = \gamma_{AB} H_{AB} + \gamma_{BA} H_{BA}$ as defined above. Then, given that we can perform unitaries of the form $U_1 = e^{-iH_{Z1}t}$ then it's possible to perform unitaries of the form $U_{A}= {-iH_{A}t}$ and $U_{B}= {-iH_{B}t}$ if $\omega_A$ and $\omega_B$ are rationally independant. An analogous result holds for $H_{AB}$
\end{proposition}
\begin{proof}
Consider $t=\frac{2\pi}{\gamma_{A}}$ and notice that 
$$U_1 = e^{-iH_{Z1}t}=e^{-i( \omega_A H_A + \omega_B H_B)t}= e^{-i2\pi H_{A}} e^{-i\frac{2\pi \omega_B}{\omega_A}H_{B}} = e^{-i\frac{2\pi \omega_B}{\omega_A}H_{B}}  $$

Since $\omega_A$ and $\omega_B$ are rationally independent, then we can generate the unitary $U_B$ and by the same argument we can generate $U_A$. Same proof applies to  $H_{Z2}$. \qed
\end{proof}

\section{Proof for universality on square Grid}\label{sec:appendix-grid}

Here we include the proofs for Sect. \ref{sec:universality-grid}.

\grid*
\begin{proof}
To show this, we can apply commutators over the available operators and obtain the two body terms and one body terms required. This can be done in a purely algebraic way, but its also useful to relate this algebraic operations to operations over the graph. First the algebraic proof is given, and then we will relate it to operations over the graph. 

Let us begin with the fact that  $[Z_{1} Z_{N+1}, X_{Grid}] = Y_{1} Z_{N+1}$ (up to global phase). Also $[Y_{1} Z_{N+1}, Z_{1} Z_{N+1}] = X_1$ and thus we can also generate $Z_1$.

Now note
\begin{equation}
    [Y_{1} Z_{N+1}, H_{A2}] = Z_{N+1}X_{1}Z_{n+1} + Z_{N+1}X_{1}Z_{2}
\end{equation}

Then since we can generate $Y_1$, we can also generate $ Z_{N+1}Z_{1}Z_{n+1} + Z_{N+1}Z_{1}Z_{2}$. Thus we have $[Z_{N+1}Y_{1} , Z_{N+1}Z_{1}Z_{n+1} + Z_{N+1}Z_{1}Z_{2}]= X_{1}Z_{n+1} + X_{1}Z_{2}$ And then we can generate $Z_{1}Z_{n+1} + Z_{1}Z_{2}$.

Note that we have now generated a Hamiltonian $H^{(2)} = Z_{1}Z_{n+1} + Z_{1}Z_{2}$ that corresponds to edges $(1,2)$ and $(1,6)$.

We will use a similar procedure to prepare Hamiltonians of the form $H^{(k)}=Z_{k-1}Z_k + R^{(k)}$, where $R^{(k)}$ does not contain the operator $Z_k$. In this way when we generate $H^{(n)}$, we will commute it with $Z_{n} Z_{N+2}$ in order to generate $Z_{k-1}Z_k$, starting from this we will be able to generate all two body terms for the first line of the form $Z_j Z_{j+1}$.

We proceed by induction, assume that we have a Hamiltonian
\begin{equation}
    H^{(k)}=Z_{k-1}Z_k + R^{(k)}
\end{equation}
where $R^{(k)}$ does not have any terms with operators $Z_k$, nor neighbours $Z_{n+k}$ and $Z_{k+1}$ and also in any vertex on the line from $k$ to $n$ (Same for $Y$ and $X$ operators). Actually it doesn't contain operators from the $k$-th column to the $n$-th.

Note that we assume that there is a vertex numbered $k+1$. We also assume we have an operator

\begin{equation}
    H^{(k)}_{A1} = X_k + X_{k+1} + ... + X_{n} + X_{R}^{(k)}
\end{equation}

Where $X_{R}^{(k)}$  is an operator without terms containing operators with support in the neighbours of vertex $k$ and also in any vertex on the line from $k$ to $n$, as before, we assume also that it doesn't contain operators from the $k$-th column to the $n$-th.
Note that
\begin{equation}
\begin{split}
[H^{(k)}, H^{(2)}_{A1}] &= [Z_{k-1}Z_k + R^{(k)}, X_k + X_{k+1} + ... + X_{n} + X_{R}^{(k)}]\\
                        &= Z_{k-1} Y_k + [R^{(k)}, X_{R}^{(k)}]
\end{split}
\end{equation}
Where now $[R^{(k)}, X_{R}^{(k)}]$ doesn't contain operators from the $k$-th column to the $n$-th. Note that $Z_{k-1}$ was not affected by the operation since $X_{R}^{(k)}$ does not have support over vertex k-1.
Perform now the operation 
\begin{equation}
\begin{split}
[[H^{(k)}, H^{(k)}_{A1}], H^{(k)}] &=[Z_{k-1} Y_k + [R^{(k)}, X_{R}^{(k)}], Z_{k-1}Z_k + R^{(k)}]\\
&= X_k + [[R^{(k)}, X_{R}^{(k)}], R^{(k)}]
\end{split}
\end{equation}
Where $[[R^{(k)}, X_{R}^{(k)}], R^{(k)}]$ has no support from column $k$ to column $n$. Define 

\begin{equation}
    \begin{split}
        H_{A1}^{(k+1)} &= H_{A1}^{(k)} - X_k - [[R^{(k)}, X_{R}^{(k)}], R^{(k)}]\\
        &=  X_k + X_{k+1} + ... + X_{n} + X_{R}^{(k)}  - X_k - [[R^{(k)}, X_{R}^{(2)}]\\
        &= X_{k+1} + ... + X_{n} + X_{R}^{(k+1)}
    \end{split}
\end{equation}
Where now  $X_{R}^{(k+1)}$ doesn't have support on $k+1$ or neighbours or from column $k+1$ to $n$.

Assume as well there's an operator $H^{(k)}_{A2}$. This operator has terms $Z_{k}Z_{neigh(k)}$ (except $Z_{k-1} Z_k$), any other term doesn't have support in $k$ or its neighbours.

Notice now that 

\begin{equation}
\begin{split}
    [[H^{(k)}, H^{(k)}_{A1}], H^{(k)}_{A2}] &= [Z_{k-1} Y_k + [R^{(k)}, X_{R}^{(2)}], H^{(k)}_{A2} ] \\
    &= Z_{k-1} X_k Z_{k+1} + Z_{k-1} X_k Z_{k+n} + [[R^{(2)}, X_{R}^{(2)}],H^{(k)}_{A2} ]
\end{split}
\end{equation}

Where $[[R^{(2)}, X_{R}^{(2)}],H^{(k)}_{A2} ]$ has has no support on $k$, $k+1$ or from columns $k+1$ to $n$.

Now consider the commutator 

\begin{equation}
\begin{split}
    &[[H^{(k)}, H^{(2)}_{A1}], Z_{k-1} X_k Z_{k+1} + Z_{k-1} X_k Z_{k+n} + [[R^{(2)}, X_{R}^{(2)}],H^{(k)}_{A2} ] ] \\
    &= [Z_{k-1} Y_k + [R^{(2)}, X_{R}^{(2)}], Z_{k-1} X_k Z_{k+1} + Z_{k-1} X_k Z_{k+n} + [[R^{(2)}, X_{R}^{(2)}],H^{(k)}_{A2} ] ]]\\
    &=  Z_k Z_{k+1} + Z_k Z_{k+n} + \overline{R}^{(k+1)}\\
    &= Z_k Z_{k+1} + R^{(k+1)} \\
    &= H^{(k+1)}
\end{split}
\end{equation}

Where $\overline{R}^{(k+1)}$  and $R^{(k+1)}$ has no support from column $k+1$ to $n$. 

Finally define $H_{A2}^{(k+1)} = H_{A2}^{(k)} - H^{(k+1)}$ and we have all the Hamiltonians necessary for step $k+1$ with the necessary properties.

We can continue this procedure until generating Hamiltonian $H^{(n)} = Z_{n-1}Z_n + R^{(n)}$. Recall that $R^{(n)}$ has no support on column $n$ of the grid.

Now note that 

\begin{equation}
    [Z_n Z_{N+2}, X_{Grid}] = Y_n Z_{N+2}
\end{equation}
And we can thus generate $X_n$.
Commuting this with $H^{(n)}$ gives $Z_{n-1}X_n Z_{N+2}$ and commuting again with $Y_n Z_{N+2}$ we obtain $Z_{n-1}Z_{n}$. We can now repeat this process with $H_{n-1}$ and $Z_{n-1}Z_{n}$. In this way we can generate all the single Pauli operators on the line ${1,..,n}$ and also the two body operators of the form $Z_{k}Z_{k+1}$ on the line. \qed
\end{proof}



%
%








\end{document}